\documentclass[draftcls,onecolumn,12pt]{IEEEtran}

\usepackage{graphicx}
\usepackage{amsmath}
\usepackage{amsfonts}
\usepackage{subfigure}
\usepackage{color}
\usepackage{algorithm2e}

\title{Incomplete decode-and-forward protocol using distributed space-time block codes}

\author{
\authorblockN{Charlotte Hucher$^1$, Ghaya Rekaya-Ben Othman$^1$ and Ahmed Saadani$^2$\\ }
\authorblockA{$^1$ Ecole Nationale Superieure des Telecommunications, Paris\\ $^2$ France Telecom Research\&Developpement \\ Email: \{hucher,rekaya\}@enst.fr, ahmed.saadani@orange-ftgroup.com}
}

\begin{document}

\maketitle

\begin{abstract}
In this work, we explore the introduction of distributed space-time codes in decode-and-forward (DF) protocols. A first protocol named the Asymmetric DF is presented. It is based on two phases of different lengths, defined so that signals can be fully decoded at relays. This strategy brings full diversity but the symbol rate is not optimal. To solve this problem a second protocol named the Incomplete DF is defined. It is based on an incomplete decoding at the relays reducing the length of the first phase. This last strategy brings both full diversity and full symbol rate.
The outage probability and the simulation results show that the Incomplete DF has better performance than any existing DF protocol and than the non-orthogonal amplify-and-forward (NAF) strategy using the same space-time codes. Moreover the diversity-multiplexing gain tradeoff (DMT) of this new DF protocol is proven to be the same as the one of the NAF.
\end{abstract}


\section{Introduction}

Diversity techniques have been developed in order to combat fading on wireless channels. Recently, a new diversity technique has been proposed with cooperative systems~\cite{coopdiv1},~\cite{coopdiv2}. Different nodes in the network cooperate in order to form a MIMO system array and exploit space-time diversity. Cooperation protocols have been classified in three main families: amplify-and-forward~(AF), decode-and-forward~(DF), and compress-and-forward~(CF).

DF protocols require more processing than AF ones, as the signals have to be decoded at relay before being forwarded. However, if signals are correctly decoded at relays, performance are better than those of AF protocols, as noise is deleted.

Moreover, in this paper, our work is motivated by the potential advantages of DF protocols over AF protocols in some scenarios. For example, it has been proven in~\cite{multihop} that in a multihop context it is necessary to use a DF protocol at some relays to regenerate the signals. Indeed a full AF strategy would add more noise at each hop, which makes signals no longer decodable.

There are few proposed DF protocols in literature. They usually do not succeed to bring both full diversity and full symbol rate. The LTW DF (named by its authors Laneman, Tse and  Wornell \cite{ltw}) has a full diversity order but a rate of $\frac{1}{2}$ symbol per channel use (symb. pcu). The NBK DF (named by its authors Nabar, Bolcskei and Kneubuhler \cite{nbk}) has a rate of 1 symb. pcu but no diversity. Indeed, as signals have to be fully decoded at relays, the first phase of the transmission needs 1 channel use for each information symbol. Two different cases can be implemented in the second phase: either the source sends new symbols to have a rate of 1 symb. pcu, but diversity is lost (these new symbols not being relayed); or the source sends the same symbols or a combination of them to have diversity, but the rate drops. The only proposed solution to this problem is the Dynamic DF (DDF) protocol \cite{ddf} which succeeds to bring both full diversity and a rate of 1 symb. pcu. However its implementation is quite complex and an usable DDF was not proposed.

To define a DF protocol with both full rate and full diversity, we suggest to introduce distributed space-time block codes (STBC) in the same way they have been successfully used in AF strategies, and in particular with the non-orthogonal AF (NAF) \cite{ddf,goldennaf}, as well as in the Alamouti DF protocol \cite{isit07}. We first present a DF protocol with asymmetric sending and relaying phases. It brings full diversity, but the rate is only $\frac{2}{3}$ symb. pcu. To solve this problem of low symbol rate, we define an Incomplete DF based on an incomplete decoding at the relays. This protocol brings both full rate and full diversity. Outage probabilities calculations and simulation results have been conducted to validate these approaches and to prove that Incomplete DF has better performance than any existing DF protocol and than the NAF using the same STBC. Moreover a theoretical study shows that the Incomplete DF has the same diversity-multiplexing gain tradeoff (DMT) than the NAF.

\section{System model and notations}

We consider a wireless network with $N+1$ sources and one destination. As the channel is shared in a TDMA manner, each user is allocated a different time slot, and the system can be reduced to a relay channel with one source, $N$ relays and one destination. The $N+1$ sources will play the role of the source in succession, while the others will be used as relays.

The channel links are assumed to be Rayleigh distributed and slow fading, so their coefficients can be considered as constant during the transmission of at least one frame. Besides, we suppose a symmetric scenario, i.e. all the channel links are subject to the same average signal-to-noise ratio (SNR).

As this work focuses on the protocol, for simplicity, a uniform energy distribution is assumed.

Considered terminals are half-duplex; they cannot receive and transmit at the same time. They are equipped with only one antenna; the MIMO case is not considered in this work.

In the next sections, notation given on figure~\ref{model} will be used. The channel coefficient of the link between source S and destination D is $g_0$, the one between source S and relay RS$_n$, $n\in\{1,\dots,N\}$, is $h_n$ and the one between relay RS$_n$ and destination D is $g_n$.

There is no channel state information (CSI) at the source, the destination is supposed to know all the channel coefficients $g_n$, which is necessary for the decoding of the information, and each relay RS$_n$ is assumed to know its corresponding source-relay channel coefficient $h_n$.

In the paper, following notation are used. Boldace lower case letters $\mathbf{v}$ denote vectors. Boldface capital letters $\mathbf{M}$ denote matrices. $\mathbf{M}^\dagger$ denote the transpose conjugate of matrix $\mathbf{M}$. $Pr$ stands for a probability. $\mathbb{R}$, $\mathbb{C}$, $\mathbb{Q}$ and $\mathbb{Z}$ stands for the real, complex, rational and integer field respectively. For each algebraic number field $\mathbb{K}$, the ring of integers is denoted $\mathcal{O}_\mathbb{K}$.

\section{The Asymmetric DF protocol}

The Asymmetric DF is a first approach to the introduction of distributed space-time codes in DF protocols. It is composed of 2 phases of different lengths. During the first phase, the source sends all the information symbols in a non-coded manner, in order for the relays to be able to easily decode them. The space-time codeword is then reconstructed in the second phase.
This protocol is to be associated with a $2N\times2N$ algebraic ST code.

\subsection{Transmission scheme}\label{sect_trans}

Let's consider the $2N\times2N$ algebraic ST code $\mathcal{C}$ which can be either a Threaded Algebraic Space-Time (TAST) code~\cite{tast}, a perfect code \cite{perfectcode} or quasi-perfect codes~\cite{quasiperfectcode}. This families of codes have a codeword which can be written in the following form 
\begin{equation}
\mathbf{X} = \left[ \begin{array}{ccccc}
x_1 & x_2 & \dots & x_{2N-1} & x_{2N} \\
\gamma\sigma(x_{2N}) & \sigma(x_{1}) & \dots & \sigma(x_{2N-2}) & \sigma(x_{2N-1}) \\
\vdots & \vdots & \ddots & \vdots & \vdots \\
\gamma\sigma^{2N-2}(x_{3}) & \gamma\sigma^{2N-2}(x_{4}) & \dots & \sigma^{2N-2}(x_{1}) & \sigma^{2N-2}(x_{2}) \\
\gamma\sigma^{2N-1}(x_{2}) & \gamma\sigma^{2N-1}(x_{3}) & \dots & \gamma\sigma^{2N-1}(x_{2N}) & \sigma^{2N-1}(x_{1})
\end{array} \right]
= \left[ \begin{array}{c} \mathbf{l_1} \\ \mathbf{l_2} \\ \vdots \\ \mathbf{l_{2N-1}} \\\mathbf{l_{2N}} \end{array} \right]
\label{eq:codeword}
\end{equation}
where the $x_k$, $k\in\{1,\dots,2N\}$, are elements of the ring of integers $\mathcal{O}_\mathbb{K}$ of $\mathbb{K}$, a cyclic extension field of $\mathbb{Q}(i)$ of dimension $2^{2N}$ (the $X_k$ are linear combinations of $2N$ information symbols), $\sigma$ is the generator of the Gallois group $\mathbb{K}/\mathbb{Q}(i)$ and $\gamma$ is an element of either $\mathbb{K}$ or $\mathbb{Z}(i)$ used to separate the layers of the codeword. Overall $4N^2$ information symbols are send in the codeword.

Let's call $\mathbf{l_k}$, $k\in\{1,\dots,2N\}$, the lines of the codeword matrix.

The transmission frame for a $N$-relay channel is described in figure~\ref{TRANS_SCHEME}. The transmission of one frame lasts $2N\times2N+2N\times N=6N^2$ channel uses. There are two main phases: during the first one, which lasts $2N\times2N=4N^2$ channel uses, the source sends the $4N^2$ symbols and the $N$ relays listen. During the second phase, which lasts $2N\times N=2N^2$ channel uses, the source sends the $N$ last lines of the codeword, while the $N$ relays send the reconstructed version of the $N$ first lines. Relay $RS_n$, $n\in\{1,\dots,N\}$, sends the recoded version of the $n^{\textrm{th}}$ line $\widetilde{\mathbf{l_n}}$ of the codeword while source sends the $(N+n)^{\textrm{th}}$ line $\mathbf{l_{N+n}}$. The destination keeps listening during the whole transmission. The symbol rate is then $\frac{4N^2}{6N^2}=\frac{2}{3}$ symb. pcu.

\subsection{Selection between the Asymmetric DF protocols and the non-cooperative case}\label{sect_selection}

DF protocols assume that signals are correctly decoded at relays during the first phase of the transmission, which is obviously not always the case. That is why we have to guarantee the first phase of the transmission to be able to use DF protocols. In literature, a selection based on the source-relay links quality was made \cite{ltw2}. The used criterion is the outage probability.

Indeed, according to Shannon theorem, if the link between source and relay $RS_n$, $n\in\{1,\dots,N\}$ is in outage, no detection is possible at this relay without error. In the other case, the source-relay $RS_n$ link is not in outage, detection is possible and we use a DF protocol assuming that no error occurs at relay $RS_n$.

In our case, the outage event of a source-relay $RS_n$ link is defined by
\begin{displaymath}
O = \left\lbrace \log\left(1+\rho|h_n|^2\right) < \frac{3}{2}R \right\rbrace
\end{displaymath}
where $\rho$ defined such that $\mathsf{SNR}(dB)=10\log_{10}(\rho)$ is the signal to noise ratio, $R$ is the global spectral efficiency, and so $\frac{3}{2}R$ is the spectral efficiency of the source-relay $RS_i$ link.

Only relays that can decode correctly the signals (whose source-relay link is not in outage) are selected. If there are $N_u\geq1$ of them selected, a DF protocol with $N_u$ relays is used, and if none of them is good, we use a non-cooperative strategy.

In practice, each relay can determine whether its source-relay link is in outage or not and send this information to the destination. The destination then knows how many relays can be used and so the scheme to be applied. The destination broadcasts this information to the other nodes of the network. This implementation aspect (channel estimation and feedback) will not be detailed any more in this paper as we focus only on the protocol.

\subsection{Performance of the Asymmetric DF protocol from simulation results}

Simulation have been made to compare the performance of the NAF and Asymmetric DF schemes in the one-relay case. Both protocols have been implemented with a distributed Golden code~\cite{goldencode} and decoded with a sphere decoder. A more detailed presentation of the Golden code is made in subsection \ref{sect_ex_golden}.

The NAF is proposed in \cite{ddf}. The protocol is non-orthogonal: the source and the relay transmit in the same time. Implemented with the distributed Golden code, the scheme is the following: the source first sends coded signals $\alpha x_1$ and $\alpha x_2$ (defined in equation (\ref{codeword})) while the relay listens. The source then sends the coded signals $i\sigma(\alpha)\sigma(x_2)$ and $\sigma(\alpha)\sigma(x_1)$ while the relay forwards the received signals.

The Asymmetric DF is implemented in the way described in figure \ref{TRANS_SCHEME}. The source first sends the information symbols $s_1$, $s_2$, $s_3$ and $s_4$ while the relay listens and decodes them. The source then sends the coded signals $i\sigma(\alpha)\sigma(x_2)$ and $\sigma(\alpha)\sigma(x_1)$ while the relay sends the coded signals $\alpha\widetilde{x_1} = \alpha(\widetilde{s_1}+\theta\widetilde{s_2})$ and $\alpha\widetilde{x_2} = \alpha(\widetilde{s_3}+\theta\widetilde{s_4})$ reconstructed from the decoded information symbols.

On figure \ref{wer_as_df} are represented the frame error rates of the SISO, NAF and Asymmetric DF protocols as functions of the SNR, for a spectral efficiency of 4 bits pcu. Even if the Asymmetric DF brings full diversity, there is a significant loss in performance (more than 5~dB) compared to the NAF protocol. This is due to the low rate (only $\frac{2}{3}$ symb. pcu) of the DF. Due to this loss, the Asymmetric DF protocol brings an advantage on the non-cooperative protocol only for SNR greater than 35~dB, which makes it useless in most cases.

\section{Incomplete DF protocol}

In order to solve the problem of low rates of the Asymmetric DF protocol, we define a new protocol named Incomplete DF. To increase the rate, the first phase of the transmission is shorten and the second phase is kept the same.
The Incomplete DF protocol is also designed to be used with a $2N\times2N$ algebraic ST code, with $N$ the number of relays. In the following, the general case is studied and two examples for the 1-relay and 2-relay cases are given.

\subsection{Transmission scheme}\label{sect_trans2}

Let's consider the same $2N\times2N$ algebraic ST code $\mathcal{C}$ to be implemented as in subsection \ref{sect_trans}, which can be for example a perfect code or a TAST code.

For the $N$-relay channel, the transmission frame is defined as described in figure~\ref{TRANS_SCHEME2}. It lasts $2N\times2N=4N^2$ channel uses and is divided in two main phases: during the first one, which lasts $2N\times N=2N^2$ channel uses, the source sends the $N$ first lines of the codeword matrix in succession and the $N$ relays listen. During the second phase, which also lasts $2N\times N=2N^2$ channel uses, the source sends the $N$ last lines of the codeword, while the $N$ relays send the decoded version of the $N$ first lines. Relay $RS_n$, $n\in\{1,\dots,N\}$, sends the decoded version of the $n^{\textrm{th}}$ line $\widetilde{\mathbf{l_n}}$ of the codeword while source sends the $(N+n)^{\textrm{th}}$ line $\mathbf{l_{N+n}}$. The destination keeps listening during the whole transmission. The symbol rate is then $\frac{4N^2}{4N^2}=1$ symb. pcu.

Received signals at destination can be expressed as in a MIMO system:
\begin{displaymath}
\left[ \begin{array}{c}
\mathbf{y_1} \\ \vdots \\ \mathbf{y_N} \\ \mathbf{y_{N+1}} \\ \vdots \\ \mathbf{y_{2N}}
\end{array} \right]
= \sqrt{\rho} \left[ \begin{array}{ccc ccc}
g_0 & \cdots & 0 & 0 & \cdots & 0 \\
\vdots & \ddots & \vdots & \vdots & \ddots & \vdots \\
0 & \cdots & g_0 & 0 & \cdots & 0 \\
\frac{g_1}{\sqrt{2}} & \cdots & 0 & \frac{g_0}{\sqrt{2}} & \cdots & 0 \\
\vdots & \ddots & \vdots & \vdots & \ddots & \vdots \\
0 & \cdots & \frac{g_N}{\sqrt{2}} & 0 & \cdots & \frac{g_0}{\sqrt{2}} \\
\end{array} \right]
\left[ \begin{array}{c}
\mathbf{l_1} \\ \vdots \\ \mathbf{l_N} \\ \mathbf{l_{N+1}} \\ \vdots \\ \mathbf{l_{2N}}
\end{array} \right]
+\left[ \begin{array}{c}
\mathbf{w_1} \\ \vdots \\ \mathbf{w_N} \\ \mathbf{w_{N+1}} \\ \vdots \\ \mathbf{w_{2N}}
\end{array} \right],
\end{displaymath}
where $\forall k \in \{1,\dots,2N\}$
\begin{itemize}
 \item $\mathbf{y_k}$ is the $k^{th}$ array of length $2N$ of the received signals,
 \item $\mathbf{l_k}$ is the $k^{th}$ line of the considered codeword matrix as defined in equation (\ref{eq:codeword}),
 \item $\mathbf{w_k}$ is an array of length $2N$ of AWGN.
\end{itemize}
The factor $\frac{1}{\sqrt{2}}$ in the channel matrix comes from the power normalization during the second transmission phase. As two terminals send in each time slot, they have to share the resources.

Reordering the received signals at destination we obtain the equivalent expression:
\begin{equation}
 \mathbf{Y_{eq}} = \sqrt{\rho} \mathbf{H_{eq}} \mathbf{X_{eq}} + \mathbf{W_{eq}}
\end{equation}
with
\begin{equation}
\mathbf{H_{eq}} = \left[ \begin{array}{cccc}
 \mathbf{H_1} & 0            & \cdots & 0      \\
 0            & \mathbf{H_2} & \cdots & 0      \\
 \vdots       & \vdots       & \ddots & \vdots \\
 0            & 0            & \cdots & \mathbf{H_{N}}
\end{array} \right]
\label{eq:Heq}
\end{equation}
and $\forall n \in \{1,\dots,N\}$
\begin{equation}
 \mathbf{H_n} = \left[ \begin{array}{cc}
 g_0 & 0 \\
 \frac{g_n}{\sqrt{2}} & \frac{g_0}{\sqrt{2}}
\end{array} \right].
\label{eq:Heqi}
\end{equation}

Decoding at destination can be performed by using ML lattice decoders such as a Schnorr-Euchner or a sphere decoder.

\subsection{Partial decoding at the relays}

The challenge of the new transmission scheme is decoding at relays. Indeed, the use of a full decode-and-forward strategy would mean that relays have to decode every information symbol $s_j$, $j\in{1,4N^2}$ of our original constellation from only $2N\times N=2N^2$ received signals.

The idea of the Incomplete DF is to estimate received signals as elements $x_k\in\mathcal{O}_\mathbb{K}$, $k\in\{1,\dots,2N\}$, without stating definitely about the information symbols $s_j$, $j\in\{1,\dots,4N^2\}$. Indeed, the knowledge of the $s_j$ is not necessary at relays, as soon as they know the signals $x_k$ that have to be forwarded. Partial decoding at relays is sufficient.

The partial decoding will be more detailed and explained in the sequel by considering some examples.

\subsection{Selection between the Incomplete DF protocols and the non-cooperative case}\label{sect_selection2}

The same selection strategy as for the Asymmetric DF (described in subsection \ref{sect_selection}) is used. Only the expressions of the outage probabilities of the source-relay links change.

Here the outage probability of a source-relay $RS_n$ link, $n\in\{1,\dots,N\}$, is given by:
\begin{equation}
Pr_{O} = Pr \left\lbrace \log(1+\rho|h_n|^2) < 2R \right\rbrace
\label{outage_event}
\end{equation}
where R is the global spectral efficiency. The spectral efficiency of the source-relay link is twice since the same information is sent in two times less channel uses.

\section{Theoretical study of the Incomplete DF performance}
\subsection{Outage probability}\label{sect_pout2}

Outage probability is given by the formula:
\begin{displaymath}
Pr_{\mathrm{out}} = Pr \left\lbrace C(\mathbf{H}) < R \right\rbrace
\end{displaymath}
with the instantaneous capacity
\begin{displaymath}
C(\mathbf{H}) = \frac{1}{T} \log\det(\mathbf{I}+\rho \mathbf{H} \mathbf{H}^\dagger)
\end{displaymath}
where $T$ is the number of time slots, $\rho$ is the signal-to-noise ratio, $\mathbf{H}$ is the (equivalent) channel matrix of the considered system and $R$ is the spectral efficiency.

\newtheorem{theorem}{Theorem}
\begin{theorem}
The outage probability of the Incomplete DF is
\begin{equation}
Pr_{\mathrm{out}} = \sum_{N_u=0}^N \binom{N}{N_u} Pr_{\mathrm{out},N_u}Pr_{O,N-N_u}.
\label{eq:pout}
\end{equation}
where $N$ is the number of relays in the network, $N_u$ is th number of relays whose source-relay link is not in outage, $Pr_{out,N_u}$ is the outage probability of the DF strategy with $N_u$ decoding relays and $Pr{O,N_{N-N_u}}$ is the probability that the source-relay links of the $N-N_u$ other relays are in outage.
\label{th:pout}
\end{theorem}

\begin{proof}
See Appendix \ref{proof:pout}.
\end{proof}

We plotted the outage probabilities thanks to Monte Carlo simulations. A relay selection has been added for all coooperative schemes. For example, in the one-relay case, the relay is chosen as the best of 3 reachable relays; in the two-relay case, relays are the two best ones between four reachable relays. In a first step the two best relays are selected, and in a second step, the DF protocol determines which of these two relays can be used (i.e. source-relay link not in outage) and chooses the corresponding strategy (SISO, Incomplete DF with 1 relay or Incomplete DF with 2 relays). The AF protocol always use both relays.

Figure \ref{pout_1relai} represents the outage probabilities of the SISO, NAF and Incomplete DF protocols as functions of the SNR at spectral efficiencies of 2 and 4 bits per channel use in the 1-relay case. We can remark that the new DF brings a slight gain over the NAF. Moreover and more interesting is the fact that due to selection it has good performance at low SNR. Same remarks can be done in the 2-relay case.

\subsection{Diversity-Multiplexing gain Tradeoff (DMT)}\label{sect_dmt2}

The diversity-multiplexing gain tradeoff (DMT) has been introduced in \cite{book_tse} to evaluate the asymptotic performance of space-time codes.

\newtheorem{definition}{Definition}
\begin{definition}
A diversity gain $d^*(r)$ is achieved at a multiplexing gain $r$ if
\begin{displaymath}
\lim_{\rho\rightarrow\infty} \frac{\log Pr_{out}(r\log\rho)}{\log\rho} = - d^*(r).
\end{displaymath}
\end{definition}

\vspace{.5cm}
\begin{theorem}
The DMT of the Incomplete DF is
\begin{equation}
 d^*(r) = (1-r)^+ + N(1-2r)^+,
\end{equation}
where $N$ is the number of relays in the network.
\label{th:dmt}
\end{theorem}

\begin{proof}
See Appendix \ref{proof:dmt}.
\end{proof}

On figure \ref{dmt} is represented the DMT for the 2-relay case. One can remark that the DMT of the Incomplete DF protocol is exactly the same as the one of the NAF protocol, outperforming the ones of the LTW and NBK DF protocols. The DMT of the DDF protocol is still better, but Incomplete DF implementation is much easier.

\section{Examples of Incomplete DF implementation and simulation results}
\subsection{1-relay channel with the Golden code}\label{sect_ex_golden}

The Golden code is an algebraic code designed for a $2\times 2$ MIMO system in \cite{goldencode} based on the cyclic division algebra of dimension 2, $\mathcal{A} = (\mathbb{Q}(i,\theta)/\mathbb{Q}(i),\sigma,\gamma)$, where $\theta = \frac{1+\sqrt{5}}{2}$ is the Golden number, $\sigma:\frac{1+\sqrt{5}}{2}\longmapsto \frac{1-\sqrt{5}}{2}$ and $\gamma=i$.

A codeword is given by
\begin{displaymath}
\mathbf{X} = \left[ \begin{array}{cc}
\alpha(s_1+\theta s_2) & \alpha(s_3+\theta s_4) \\
i\sigma(\alpha)(s_3+\sigma(\theta)s_4) & \sigma(\alpha)(s_1+\sigma(\theta)s_2)
\end{array} \right] 
\end{displaymath}
with the $s_j$, $j\in\{1,\dots,4\}$ being the information symbols taken in a QAM constellation and $\alpha = 1+i-i\theta$. The elements of the code matrix are in $\mathcal{O}_\mathbb{K}$ the ring of integers of the number field $mathbb{K}=\mathbb{Q}(i,\theta)$. Let's note them $x_1=s_1+\theta s_2$ and $x_2=s_3+\theta s_4$. The codeword is then:
\begin{equation}
\mathbf{X} = \left[ \begin{array}{cc}
\alpha x_1                 & \alpha x_2                \\
i\sigma(\alpha)\sigma(x_2) & \sigma(\alpha)\sigma(x_1)
\end{array} \right].
\label{codeword}
\end{equation}

We propose to implement this space-time code in a distributed manner using the new Incomplete DF protocol with 1 relay as described in subsection~\ref{sect_trans2}. The transmission frame is described in figure~\ref{GOLDEN_DF}.

Elements $x_1$ and $x_2$ both contain two information symbols. They have to be recovered respectively from the received signals $y^{r}_1$ and $y^{r}_2$. The idea of "incomplete decoding" is to decode $x_1$ and $x_2$ as elements $\mathcal{O}_\mathbb{K}$ without stating definitely on the information symbols.

We consider in this paper two decoding methods:
\begin{itemize}
 \item an exhaustive search,
 \item a diophantine approximation.
\end{itemize}

\paragraph{Exhaustive search}

Let's assume the information symbols $s_j$, $j\in\{1,\dots,4\}$, belong to a constellation $C$ (for example a 4-QAM constellation, see figure \ref{const}). We can define a new constellation $C'$ to which the coded symbols $x_k$, $k\in\{1,\dots,2\}$, belong (see figure \ref{const}), which is a finite subset of $\mathcal{O}_\mathbb{K}$.
An exhaustive search is performed in this new constellation. $\widetilde{x_1}$ (resp. $\widetilde{x_2}$) is obtained by looking for the element $x$ of $C'$ that minimizes the distance between $y^{r}_1$ (resp. $y^{r}_2$) and $\sqrt{\rho}h_1 x$.
\begin{align*}
\widetilde{x_1} & = \arg\min_{x \in C'} \{ |y^{r}_1-\sqrt{\rho}h_1 x|^2 \} \\
\widetilde{x_2} & = \arg\min_{x \in C'} \{ |y^{r}_2-\sqrt{\rho}h_1 x|^2 \}
\end{align*}

The complexity of the exhaustive search grows with the size of the constellation. In the case of an $M$-QAM the complexity of the exhaustive search is of the order $M^2$.
However, decomposing signals in their real and imaginary parts, complexity can be reduced to the order $M$.

Simulations have been run for the one-relay cooperative scheme with the distributed Golden code for spectral efficiencies of 2 and 4 bits per channel use. The same relay selection as for the outage probability has been applied.
Figure \ref{wer_1relai} represents the frame error rates of the SISO, NAF and new DF protocols as functions of the SNR. The good performance for low and high SNR noticed in subsection \ref{sect_pout2} on the outage probability curves are confirmed here by simulation results. Using an exhaustive decoding, we obtain slight asymptotic gains over the NAF protocol. Moreover, we can see (especially for 4 bits pcu) that the proposed DF protocol has better performance at low SNR.

\paragraph{Diophantine approximation}

In order to reduce relay decoding complexity, we propose to use a diophantine approximation of the $x_k$, $k\in\{1,\dots,2\}$. There exist two types of diophantine approximation.

\begin{definition}
A homogeneous diophantine approximation of $\zeta\in\mathbb{R}$ is a fraction $\frac{p}{q}\in\mathbb{Q}$ such that $|\zeta-\frac{p}{q}|$ or $D(p,q)=|q\zeta-p|$ is small.
\end{definition}

\begin{definition}
An inhomogeneous diophantine approximation of $\zeta\in\mathbb{R}$, given $\beta\in\mathbb{R}$, is a fraction $\frac{p}{q}\in\mathbb{Q}$ such that $D(p,q)=|q\zeta-p-\beta|$ is small.
\end{definition}

\begin{definition}
A pair $(p,q)\in\mathbb{N}^2$ is a best diophantine approximation if $\forall(p',q')\neq (p,q)\in\mathbb{N}^2$, we have:
\begin{displaymath}
q' \leq q \Rightarrow D(p',q') \geq D(p,q).
\end{displaymath}
\end{definition}

Cassels' algorithm has been proposed in \cite{cassels} and explained in details in \cite{dioph_approx}. Given $\zeta,\beta\in\mathbb{R}$, this algorithm enumerates all best inhomogeneous approximations. A simple modification of this algorithm provides $(p,q)$ in a finite set $\{1,\dots,Z\}$ that minimizes $D(p,q)$. A change of basis provides $(p,q)$ in a $Z$-PAM. The modified algorithm has a complexity of the order $\sqrt{Z}$.

Diophantine approximation only deals with real numbers. The problem of decoding at the relay has then to be divided into its real and imaginary parts. Let's note $\widetilde{y^{r}_1} = \frac{y^{r}_1}{\sqrt{\rho h_1\alpha}}$. Given $\theta,Re(\widetilde{y^{r}_1})\in\mathbb{R}$, we want to find $(Re(s_1),Re(s_2))\in\sqrt{M}$-PAM such that
\begin{displaymath}
|Re(\widetilde{y^{r}_1})-Re(s_1)-\theta Re(s_2)|
\end{displaymath}
is minimized. To solve this minimization we can use the modified Cassels' algorithm with the parameters $\beta=-(Re(\widetilde{y^{r}_1})+(\sqrt{M}+1)(1+\theta))/2$ and $\zeta=-\theta$. The final algorithm is given in appendix \ref{alg:dioph_approx}.

The same processing is done to decode the imaginary part of the signal. Finally, the decoding complexity is only $\sqrt{\sqrt{M}}=\sqrt[4]{M}$.

If $\theta=e^{i\frac{\pi}{4}}$ the decomposition in real and imaginary part is more complex, but the diophantine approximation still can be used with a slight modification of the given algorithm. Thus the diophantine approximation can also be applied when using a distributed TAST code.

Simulations have also been run with the Golden code using a diophantine approximation at the relay. We can see on figure \ref{wer_1relai} that in this case performance are slightly worse. This is explained by the non-optimal decoding at relays. However, this slight gap in performance (only 0.5 dB) is compensated by a much lower decoding complexity decreasing from $M$ to $\sqrt[4]{M}$.

\subsection{2-relay channel with the TAST code}\label{sect_ex_tast44}

For the 2-relay case, we propose to use the $4\times4$ TAST code in a distributed manner, associated to the Incomplete DF protocol. After recalling the structure of the TAST code, we will introduce two partial decoding methods, and so justify our choice of the TAST code.

TAST codes, introduced in \cite{tast}, are layered space-time codes. Here we use the $4\times4$ TAST code constructed using the cyclotomic field $\mathbb{K}=\mathbb{Q}(i,\theta)$, where $\theta=e^{i\frac{\pi}{8}}$, the generator of the Gallois group $\sigma:\theta\longmapsto i\theta$ and $\phi=e^{i\frac{\pi}{8}}$. The codeword is
\begin{displaymath}
\mathbf{X} = \left[ \begin{array}{cccc}
x_1 & x_2 & x_3 & x_4 \\
\phi\sigma(x_4) & \sigma(x_1) & \sigma(x_2) & \sigma(x_3) \\
\phi\sigma^2(x_3) & \phi\sigma^2(x_4) & \sigma^2(x_1) & \sigma^2(x_2) \\
\phi\sigma^3(x_2) & \phi\sigma^3(x_3) & \phi\sigma^3(x_4) & \sigma^3(x_1)
\end{array} \right],
\end{displaymath}
where, $\forall k\in\{1,\dots,4\}$, $x_k = s_{4*k-3} + \theta s_{4*k-2} + \theta^2 s_{4*k-1} + \theta^3 s_{4*k}$.

We propose to use this code in a 2-relay cooperative system as described in subsection \ref{sect_trans2}. The transmission scheme is schematized in figure~\ref{TAST_DF}.

Elements $x_1$, $x_2$, $x_3$ and $x_4\in\mathcal{O}_\mathbb{K}$ and their conjugates have to be recovered from the signals $y^{r_j}_1$ to $y^{r_j}_8$ received at the relay RS$_j$, $j\in\{1,\dots,2\}$.
We propose here two different methods for the partial decoding.

\paragraph{Exhaustive search (dimension 4)}

$x_k$, $k\in\{1,\dots,4\}$ and their conjugates $\sigma(x_n)$ can be decoded at relays by an exhaustive search as in the 1-relay case. The difference is that $x_n$ and $\sigma(x_n)$ cannot be decoded separately as they are conjugates.

Let's assume the $s_j$, $j\in\{1,\dots,16\}$, belong to a constellation $C$. We can define a new constellation $C_1$ to which the $x_k$ belong, and a corresponding constellation $C_2$ to which their conjugates $\sigma(x_k)$ belong.

Decoded versions of the $x_k$ and their conjugates $\sigma(x_k)$ are obtained by looking for the elements $x$ of constellation $C_1$ and $\sigma(x)$ of constellation $C_2$ minimizing the distance between $\sqrt{\rho}h_1 x$ and the received signal corresponding to $x_k$ and the distance between $\sqrt{\rho}h_1\sigma(x)$ and the received signal corresponding to $\sigma(x_k)$. We decide to minimize the sum of these two distances. For example:
\begin{displaymath}
\{\widetilde{x_1},\widetilde{\sigma(x_1)}\} 
= \arg\min_{x \in C_1, \sigma(x) \in C_2} 
    \left\lbrace \left|y^{r_j}_1-\sqrt{\rho}h_1 x\right|^2 + \left|y^{r_j}_6-\sqrt{\rho}h_1\sigma(x)\right|^2 \right\rbrace
\end{displaymath}

However, this exhaustive decoding can be quite complex if a high constellation size is considered. Indeed, if the information symbols $s_j$ belong to a $M$-QAM constellation ($M$ elements), then, the $x_k$ have to be decoded in a new constellation of $M^4$ elements. The complexity of decoding is of the order $M^4$.

Simulations have been run for the two-relay cooperative scheme with a $4\times4$ perfect code for spectral efficiencies of 2 and 4 bits per channel use.
Figure \ref{wer_2relais} represents the frame error rates of the SISO, NAF and new DF protocols as functions of the SNR. The same remarks than in the one-relay case can be done. The Incomplete DF and the NAF have nearly the same performance (nearly zero asymptotic gain), but due to selection the proposed DF protocol outperforms the NAF at low SNR.

\paragraph{Two steps exhaustive decoding (dimension 2)}

A slight modification can reduce this exhaustive decoding in a constellation of $M^4$ elements to two exhaustive decodings in a constellation of only $M^2$ elements.

We can notice that $x_1$ and its second conjugate $\sigma^2(x_1)$ can be rewritten in the form:
\begin{align*}
x_1 & = (s_1 + \theta^2 s_3) + \theta (s_2 + \theta^2 s_4) \\
\sigma^2(x_1) & = (s_1 + \theta^2 s_3) - \theta (s_2 + \theta^2 s_4)
\end{align*}
\begin{equation}
\left[ \begin{array}{c}
x_1 \\ \sigma^2(x_1)
\end{array}\right]
= \underbrace{\left[ \begin{array}{cc}
1 & \theta \\ 1 & -\theta
\end{array}\right]}_{\mathbf{M}}
\left[ \begin{array}{c}
z_1 \\ z_2
\end{array}\right],
\label{eq:method2}
\end{equation}
where $z_1 = (s_1 + \theta^2 s_3)$ and $z_2 = (s_2 + \theta^2 s_4)$ are elements of the ring of integers of the field $\mathbb{Q}(e^{i\frac{\pi}{4}})$ of dimension 2 over $\mathbb{Q}(i)$. As $\frac{1}{\sqrt{2}}\mathbf{M}$ is a rotation matrix, a simple multiplication by $\mathbf{M}^\dagger$ allows to obtain $z_1$ and $z_2$ from $x_1$ and $\sigma^2(x_1)$.

In order to take advantage of this property, the idea is that the source sends the first and third lines of the codeword matrix during the first phase of the transmission and the second and fourth lines during the second phase of the transmission.

The partial decoding at relays is then done in two steps. First we compute the matrix product
\begin{displaymath}
\left[ \begin{array}{c}
z'_1 \\ z'_2
\end{array}\right]
= \frac{1}{2} \mathbf{M}^\dagger
\left[ \begin{array}{c}
\frac{y^{r_k}_1}{\sqrt{\rho}h_1} \\ \frac{y^{r_k}_6}{\sqrt{\rho}h_1}
\end{array}\right].
\end{displaymath}
Then we decode elements $z_1$ and $z_2$ of the ring of integers of $\mathbb{Q}(e^{i\frac{\pi}{4}})$ in an exhaustive way as in the example of the subsection \ref{sect_ex_golden}. Finally $x_1$ and its conjugate $\sigma^2(x_1)$ can be easily deduced from equation (\ref{eq:method2}).

This second method allows to decrease considerably the complexity. Indeed, the exhaustive search is now performed in a constellation of $M^2$ elements instead of $M^4$, which is quite reasonable.

This second decoding method cannot be applied to $4\times4$ perfect codes whose structure do not have the same property. That is why we have chosen to use TAST code. The advantages of the diophantine approximation and this two-step decoding method could be combined. Ideally the Incomplete DF would be the least complex if used with a distributed STBC offering both a structure allowing the two-step decoding and $\theta\in\mathbb{R}$ for a simple diophantine approximation.

Simulations have also been run for the two-relay cooperative scheme with a $4\times4$ TAST code and the two-step decoding at the relays. One can see on figure \ref{wer_2relais} that TAST codes provide slightly worse performance than perfect codes. This is explained by the fact that perfect codes are NVD (non-vanishing-determinant), on the contrary of TAST codes. However, when we use two or more relays, the partial decoding of the information at relays induces more complexity, as the two-step decoding method described in subsection \ref{sect_ex_tast44} cannot be used. That is why the use of two-step decodable STBC such as the TAST codes is necessary.

\section{Conclusion}

In this paper, we define a new DF protocol using distributed space-time codes that provides both full diversity and full rate, as the best known AF protocols, and unlike the existing LTW and NBK DF protocols. This new protocol is based on an incomplete decoding of the signal at the relays. The received signals at relays are decoded as elements of the ring of integers of the considered number field without decoding the information symbols. Several decoding methods are proposed at relays: exhaustive search, diophantine approximation or a method based on the decoding decomposition in two steps according to the code structure. The two last methods allow a considerable decrease of complexity.

The diversity-multiplexing gain tradeoff is proved to be the same as the one of the NAF protocol which is the best known AF protocol. Besides outage probability and simulation results prove that the Incomplete DF gives slightly better performance than the NAF protocol in the high SNR regime, and selection provides an improvement  for low SNR.

In this study, we have only considered cooperative schemes with line-of-sight, but the use of DF protocols can be very important in a non-line-of-sight or multihop cooperative scheme. To highlight the advantages of this new DF protocol over an AF strategy, applications of the Incomplete DF to a multihop system will be investigated in future works.

\appendix
\section{Appendix}
\subsection{Proof of Theorem \ref{th:pout}}\label{proof:pout}

Two cases have to be distinguished: with or without cooperation.
When $N_u \geq 1$ source-relay links are not in outage ($N-N_u$ relays are in outage) , the Incomplete DF cooperation scheme with $N_u$ relays is used.

From equation (\ref{eq:Heq}) we can write
\begin{displaymath}
\det(\mathbf{I}+\rho \mathbf{H}\mathbf{H}^\dagger) = \prod_{i=1}^{N_u}\det(\mathbf{I}+\rho \mathbf{H_i}\mathbf{H{_i}}^\dagger)
\end{displaymath}
and using equation (\ref{eq:Heqi}) the outage probability can be written
\begin{equation}
\begin{array}{rl}
Pr_{\mathrm{out},N_u} & = Pr \left\lbrace \frac{1}{2N_u} \log\left( \prod_{i=1}^{N_u} \left(1 + \frac{\rho}{2} \left( 3|g_0|^2 + |g_i|^2 \right) + \frac{\rho^2}{2} |g_0|^4 \right) \right) < R \right\rbrace \\
& = Pr \left\lbrace \log\prod_{i=1}^{N_u} \left( 1 + \frac{\rho}{2} \left( 3|g_0|^2 + |g_i|^2 \right) + \frac{\rho^2}{2} |g_0|^4 \right) < 2N_uR \right\rbrace.
\end{array}
\label{eq:pout_k}
\end{equation}

As the outage events of the source-relay links are independent, the probability of having only the last $N-N_u$ relays in outage, is the product of the probabilities of each $N_u$ first source-relay links not being in outage, and each of the last $N-N_u$ source-relay links being in outage. These probabilities are given by expression (\ref{outage_event}). So the outage probability of the last $N-N_u$ relays only can be written
\begin{equation}
Pr_{O,N-N_u} = \prod_{i=1}^{N_u} Pr \left\lbrace \log\left( 1+\rho|h_i|^2 \right) > 2R \right\rbrace \prod_{i=N_u+1}^N Pr \left\lbrace \log\left( 1+\rho|h_i|^2 \right) < 2R \right\rbrace.
\label{eq:po_k}
\end{equation}

When all source-relay links are in outage, we use the non-cooperative scheme, whose outage probability is
\begin{equation}
Pr_{\mathrm{out},0} = Pr \left\lbrace  \log\left( 1 + \rho |g_0|^2\right) < R \right\rbrace
\label{eq:pout_0}
\end{equation}

As the outage events of all source-relay links are independent, the probability of this case is given by the product of the outage probabilities of each source-relay link.
\begin{equation}
Pr_{O,N} = \prod_{i=1}^N Pr \left\lbrace \log\left( 1+\rho|h_i|^2 \right) < 2R \right\rbrace
\label{eq:po_0}
\end{equation}

Finally, as there are $\binom{N}{N_u}$ possible combinations of $N_u$ relays in $N$, we can write in the general case
\begin{equation}
Pr_{\mathrm{out}} = \sum_{N_u=0}^N \binom{N}{N_u} Pr_{\mathrm{out},N_u}Pr_{O,N-N_u}.
\end{equation}

\subsection{Proof of Theorem \ref{th:dmt}}\label{proof:dmt}
\subsubsection{Preliminaries}

\begin{definition}
Let $g$ follow a Rayleigh distribution. The exponential order of $\frac{1}{|g|^2}$ is
\begin{displaymath}
u = -\lim_{\rho\rightarrow \infty} \frac{\log|g|^2}{\log\rho}.
\end{displaymath}
\end{definition}

We can note $|g|^2 \doteq \rho^{-u}$ where the notation $\doteq$ denotes an asymptotic behavior when $\rho\rightarrow \infty$.

\newtheorem{lemma}{Lemma}
\begin{lemma}
The probability density function of $u$ is
\begin{displaymath}
p_u = \lim_{\rho\rightarrow \infty} \log(\rho) \rho^{-u} \exp(-\rho^{-u}),
\end{displaymath}
which satisfies
\begin{displaymath}
p_u \doteq \left\lbrace \begin{array}{l} \rho^{-\infty},\text{ for }u<0 \\
					 \rho^{-u},\text{ for }u\geq0 \end{array} \right.
\end{displaymath}
\end{lemma}

\begin{lemma}
Let $\mathcal{O}$ be a certain set and $P_\mathcal{O} = Pr\left\lbrace (u_1,\dots,u_N)\in\mathcal{O}\right\rbrace$, then
\begin{displaymath}
P_\mathcal{O} \doteq \rho^{-d} \text{ with } d = \inf_{(u_1,\dots,u_N)\in\mathcal{O}^+} \sum_{j=1}^N u_j
\end{displaymath}
where $\mathcal{O}^+ = \mathcal{O} \cap \mathbb{R}^{N+}$
\end{lemma}

\begin{proof}
The proof is drawn in \cite{dmt_mac}.
\end{proof}

\subsubsection{Proof of the theorem}

The outage probability of the Incomplete DF is given in equation~(\ref{eq:pout}). In order to compute the DMT of this cooperative strategy, we have to study the asymptotic behavior of this expression when $\rho$ grows to infinity.

Let $u_0$, $u_n$ and $v_n$, $n\in\{1,\dots,N\}$ be the exponential orders of $\frac{1}{|g_0|^2}$, $\frac{1}{|g_n|^2}$ and $\frac{1}{|h_n|^2}$ respectively.

In the case of signals being correctly decoded at $N_u$ relays
\begin{displaymath}
Pr_{\mathrm{out},N_u} = Pr \left\lbrace \log\prod_{i=1}^{N_u} \left( 1 + \frac{\rho}{2} \left( 3|g_0|^2 + |g_i|^2 \right) + \frac{\rho^2}{2} |g_0|^4 \right) < 2N_uR \right\rbrace
\end{displaymath}
which asymptotically becomes
\begin{align*}
Pr_{\mathrm{out},N_u} & \doteq Pr \left\lbrace \sum_{i=1}^{N_u} \log\left( \rho^{1-v_0} + \rho^{1-v_i} + \rho^{2-2v_0} \right) < 2N_ur\log\rho \right\rbrace \\
& \doteq Pr \left\lbrace \sum_{i=1}^{N_u} \max\left(1-v_i,2-2v_0\right) < 2N_ur \right\rbrace \\
& \doteq \rho^{-d_{out,N_u}(r)}
\end{align*}

As $\sum_{i=1}^{N_u} (1-v_i) < 2N_ur$ gives $N_u(1-2r)< \sum_{i=1}^{N_u} v_i$ and $\sum_{i=1}^{N_u} (2-2v_0) < 2N_ur$ gives $1-r< v_0$, we obtain the diversity-multiplexing gain tradeoff
\begin{equation}
d_{out,N_u}(r)=\inf\left(v_0+\sum_{i=1}^{N_u} v_i\right)=(1-r)+N_u(1-2r)^+.
\end{equation}

This case occurs when $N-N_u$ of the source-relay links are in outage and the others are not, with the probability:
\begin{align*}
Pr_{O,N-N_u} & = \prod_{i=1}^{N_u} Pr \left\lbrace \log\left( 1+\rho|h_i|^2 \right) > 2R \right\rbrace \prod_{i=N_u+1}^N Pr \left\lbrace \log\left( 1+\rho|h_i|^2 \right) < 2R \right\rbrace \\
& \doteq \prod_{i=1}^{N_u} \left( 1- Pr \left\lbrace \log\rho^{1-u_i} < 2r\log\rho \right\rbrace \right) \prod_{i=N_u+1}^N Pr \left\lbrace \log\rho^{1-u_i} < 2r\log\rho \right\rbrace \\
& \doteq \prod_{i=1}^{N_u} \left( 1- Pr \left\lbrace 1-u_i < 2r \right\rbrace \right) \prod_{i=N_u+1}^N Pr \left\lbrace 1-u_i < 2r \right\rbrace \\
& \doteq \prod_{i=1}^{N_u} \left( 1 - \rho^{-(1-2r)}\right) \prod_{i=N_u+1}^N \rho^{-(1-2r)} \\
& \doteq 1 \times \rho^{-(N-N_u)(1-2r)}
\end{align*}
so the diversity-multiplexing gain tradeoff is
\begin{equation}
d_{O,N-N_u}(r)=(N-N_u)(1-2r)^+.
\end{equation}

In the case of all source-relay links being in outage:
\begin{align*}
Pr_{\mathrm{out},0} & = Pr \left\lbrace \log\left( 1 + \rho |g_0|^2\right) < R \right\rbrace \\
& \doteq Pr \left\lbrace \log\rho^{1-v_0} < r\log\rho \right\rbrace
\doteq Pr \left\lbrace 1-v_0 < r \right\rbrace \\
& \doteq \rho^{-d_{out,0}(r)}
\end{align*}
with the diversity-multiplexing gain tradeoff
\begin{equation}
d_{out,0}(r)=1-r.
\end{equation}

This case occurs with the probability:
\begin{align*}
Pr_{O,N} & = \prod_{i=1}^N Pr \left\lbrace \log\left( 1+\rho|h_i|^2 \right) < 2R \right\rbrace \\
& \doteq \prod_{i=1}^N Pr \left\lbrace \log\rho^{1-u_i} < 2r\log\rho \right\rbrace \\
& \doteq \prod_{i=1}^N Pr \left\lbrace 1-u_i < 2r \right\rbrace
\doteq \prod_{i=1}^N \rho^{-(1-2r)} \\
& \doteq \rho^{-N(1-2r)}
\end{align*}
so the diversity-multiplexing gain tradeoff is
\begin{equation}
d_{O,N}(r)=N(1-2r)^+.
\end{equation}

Finally we can write
\begin{align*}
Pr_{\mathrm{out}} & = \sum_{N_u=0}^N C_{N_u}^N Pr_{\mathrm{out},N_u}Pr_{O,N-N_u} \\
& \doteq \sum_{N_u=0}^N C_{N_u}^N \rho^{-d_{out,N_u}(r)} \rho^{-d_{O,N-N_u}(r)} \\
& \doteq \rho^{-\max_{N_u \in \{0,\dots,N\}}(d_{out,N_u}(r)+d_{O,N-N_u}(r))}
\end{align*}
so the total diversity-multiplexing gain tradeoff is
\begin{equation}
d(r)=\max_{N_u \in \{0,\dots,N\}}(d_{out,N_u}(r)+d_{O,N-N_u}(r)) = (1-r)+N(1-2r)^+.
\end{equation}

\subsection{Modified Cassels' algorithm for decoding Z-PAM}\label{alg:dioph_approx}

\begin{algorithm}[H]
\SetLine
\KwIn{$y,\theta,Z$}
\KwOut{$P,Q$}
$\beta = -(y + (Z+1)(1+\theta))/2$\;
$\alpha = -\theta$\;
$D_{min} = \infty$\;
$\eta_0 = \alpha; \eta_1 = -1; \zeta = -\beta$\;
$p_0 = 0; p_1 = 1; P_1 = 0$\;
$q_0 = 1; q_1 = 0; Q_1 = 0$\;
\end{algorithm}
\begin{algorithm}[H]
\SetLine
\While{$\eta_{n-1} \neq 0 \wedge \zeta_{n-1} \neq 0 \wedge Q_{n-1} \leq Z$}{
  $a_n = \lfloor -\frac{\eta_{n-2}}{\eta_{n-1}} \rfloor$\;
  $p_n = p_{n-2} + a_n p_{n-1}$\;
  $q_n = q_{n-2} + a_n q_{n-1}$\;
  $\eta_n = \eta_{n-2} + a_n \eta_{n-1}$\;
  \eIf{$Q_{n-1} \leq q_{n-1}$}{
    $b_n = \lfloor -\frac{\zeta_{n-1}-\eta_{n-2}}{\eta_{n-1}} \rfloor$\;
    $P_n = P_{n-1} + p_{n-2} + b_n p_{n-1}$\;
    $Q_n = Q_{n-1} + q_{n-2} + b_n q_{n-1}$\;
    $\zeta_n = \zeta_{n-1} + \eta_{n-2} + b_n \eta_{n-1}$\;
  }{
    $P_n = P_{n-1} - p_{n-1}$\;
    $Q_n = Q_{n-1} - q_{n-1}$\;
    $\zeta_n = \zeta_{n-1} - \eta_{n-1}$\;
  }
  $P' = 2 P_n - (Z+1)$\;
  $Q' = 2 Q_n - (Z+1)$\;
  $D_{curr} = (y - P' - \theta Q')^2$\;
  \If{$D_{curr} \leq D_{min}$}{
    $P = P'$\;
    $Q = Q'$\;
    $D_{min} = D_{curr}$\;
  }
  $n = n+1$\;
}
\end{algorithm}

\bibliographystyle{IEEEtran}
% \bibliography{idf_journal}

\begin{thebibliography}{10}
\providecommand{\url}[1]{#1}
\csname url@rmstyle\endcsname
\providecommand{\newblock}{\relax}
\providecommand{\bibinfo}[2]{#2}
\providecommand\BIBentrySTDinterwordspacing{\spaceskip=0pt\relax}
\providecommand\BIBentryALTinterwordstretchfactor{4}
\providecommand\BIBentryALTinterwordspacing{\spaceskip=\fontdimen2\font plus
\BIBentryALTinterwordstretchfactor\fontdimen3\font minus
  \fontdimen4\font\relax}
\providecommand\BIBforeignlanguage[2]{{%
\expandafter\ifx\csname l@#1\endcsname\relax
\typeout{** WARNING: IEEEtran.bst: No hyphenation pattern has been}%
\typeout{** loaded for the language `#1'. Using the pattern for}%
\typeout{** the default language instead.}%
\else
\language=\csname l@#1\endcsname
\fi
#2}}

\bibitem{coopdiv1}
A.~Sendonaris, E.~Erkip, and B.~Aazhang, ``User {C}ooperation {D}iversity.
  {P}art {I}. {S}ystem {D}escription,'' \emph{{IEEE} Trans. Commun.}, vol.~51,
  no.~11, pp. 1927--1938, November 2003.

\bibitem{coopdiv2}
------, ``User {C}ooperation {D}iversity. {P}art {II}. {I}mplementation
  {A}spects and {P}erformance {A}nalysis,'' \emph{{IEEE} Trans. Commun.},
  vol.~51, no.~11, pp. 1939--1948, November 2003.

\bibitem{multihop}
S.~Yang and J.-C. Belfiore, ``Diversity of {M}{I}{M}{O} multihop relay
  channels,'' \emph{{IEEE} Trans. Inform. Theory}, August 2007, submitted.

\bibitem{ltw}
J.~Laneman and G.~Wornell, ``Distributed space-time coded protocols for
  exploiting cooperative diversity in wireless networks,'' \emph{{IEEE} Trans.
  Inform. Theory}, vol.~49, no.~10, pp. 2415--2425, October 2003.

\bibitem{nbk}
R.~Nabar, H.~Bolcskei, and F.~Kneubuhler, ``Fading relay channels: performance
  limits and space-time signal design,'' \emph{{IEEE} J. Select. Areas
  Commun.}, vol.~22, no.~6, pp. 1099--1109, August 2004.

\bibitem{ddf}
K.~Azarian, H.~E. Gamal, and P.~Schniter, ``On the achievable
  diversity-multiplexing tradeoff in half-duplex cooperative channels,''
  \emph{{IEEE} Trans. Inform. Theory}, vol.~51, no.~12, pp. 4152--4172,
  December 2005.

\bibitem{goldennaf}
S.~Yang and J.~Belfiore, ``Optimal {S}pace-{T}ime {C}odes for the {MIMO}
  {A}mplify-and-{F}orward {C}ooperative {C}hannel,'' in \emph{2006
  International Zurich Seminar on Communications}, February 2006, pp. 122--125.

\bibitem{isit07}
C.~Hucher, G.~R.-B. Othman, and J.-C. Belfiore, ``{AF} and {DF} {P}rotocols
  based on {A}lamouti {ST} {C}ode,'' in \emph{IEEE International Symposium on
  Information Theory}, June 2007, pp. 1526--1528.

\bibitem{tast}
H.~E. Gamal and M.~O. Damen, ``Universal {S}pace-{T}ime {C}oding,''
  \emph{{IEEE} Trans. Inform. Theory}, vol.~49, no.~5, pp. 1097--1119, May
  2003.

\bibitem{perfectcode}
F.~Oggier, G.~Rekaya, J.-C. Belfiore, and E.~Viterbo, ``Perfect {S}pace-{T}ime
  {B}lock {C}odes,'' \emph{{IEEE} Trans. Inform. Theory}, vol.~52, no.~9, pp.
  3885--3902, September 2006.

\bibitem{quasiperfectcode}
P.~Elia, K.~Kumar, S.~Pawar, P.~Kumar, and L.~Hsiao-Feng, ``Explicit
  {S}pace-{T}ime {C}odes {A}chieving the {D}iversity-{M}ultiplexing {G}ain
  {T}radeoff,'' \emph{{IEEE} Trans. Inform. Theory}, vol.~52, no.~9, pp.
  3869--3884, September 2006.

\bibitem{ltw2}
J.~Laneman, D.~Tse, and G.~Wornell, ``Cooperative diversity in wireless
  networks: {E}fficient protocols and outage behavior,'' \emph{{IEEE} Trans.
  Inform. Theory}, vol.~50, no.~12, pp. 3062--3080, December 2004.

\bibitem{goldencode}
J.-C. Belfiore, G.~Rekaya, and E.~Viterbo, ``The {G}olden {C}ode: A 2x2
  {F}ull-{R}ate {S}pace-{T}ime {C}ode with {N}on-{V}anishing {D}eterminants,''
  \emph{{IEEE} Trans. Inform. Theory}, vol.~51, no.~4, pp. 1432--1436, April
  2005.

\bibitem{book_tse}
D.~Tse and P.~Viswanath, \emph{Fundamentals of {W}ireless
  {C}ommunication}.\hskip 1em plus 0.5em minus 0.4em\relax Cambridge University
  Press, September 2004, draft to be published.

\bibitem{cassels}
J.~Cassels, \emph{An {I}ntroduction to {D}iophantine {A}pproximation}.\hskip
  1em plus 0.5em minus 0.4em\relax Cambridge University Press, 2005.

\bibitem{dioph_approx}
I.~Clarkson, ``Approximation of linear forms by lattice points with
  applications to signal processing,'' Ph.D. dissertation, Australien National
  University, 1997.

\bibitem{dmt_mac}
L.~Zheng and D.~N.~C. Tse, ``Diversity and multiplexing: a fundamental tradeoff
  in multiple-antenna channels,'' \emph{{IEEE} Trans. Inform. Theory}, vol.~49,
  no.~5, pp. 1073--1096, May 2003.

\end{thebibliography}

\newpage

\begin{figure}[h!t]
\centering
\includegraphics[width=.3\linewidth]{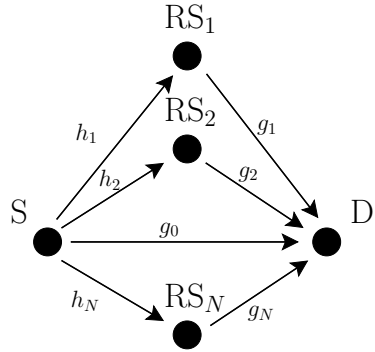}
\caption{System model : relay channel with one source, N relays and one destination} \label{model}
\end{figure}

\begin{figure}[h!t]
\centering
\includegraphics[width=\linewidth,clip]{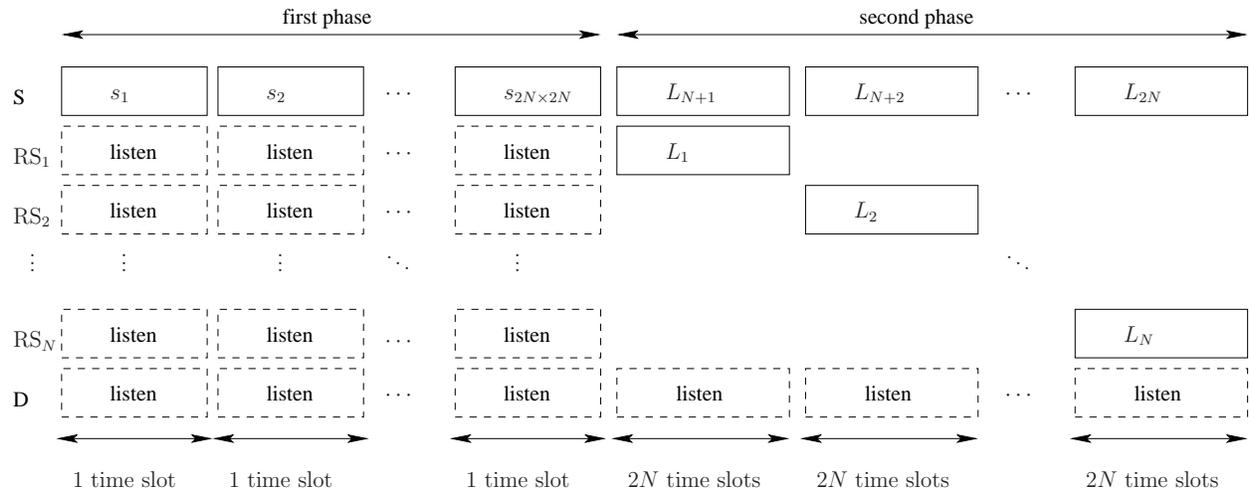}
\caption{Transmission frame of the Asymmetric DF protocol} \label{TRANS_SCHEME}
\end{figure}

\begin{figure}[h!t]
\centering
\includegraphics[width=.7\linewidth,clip]{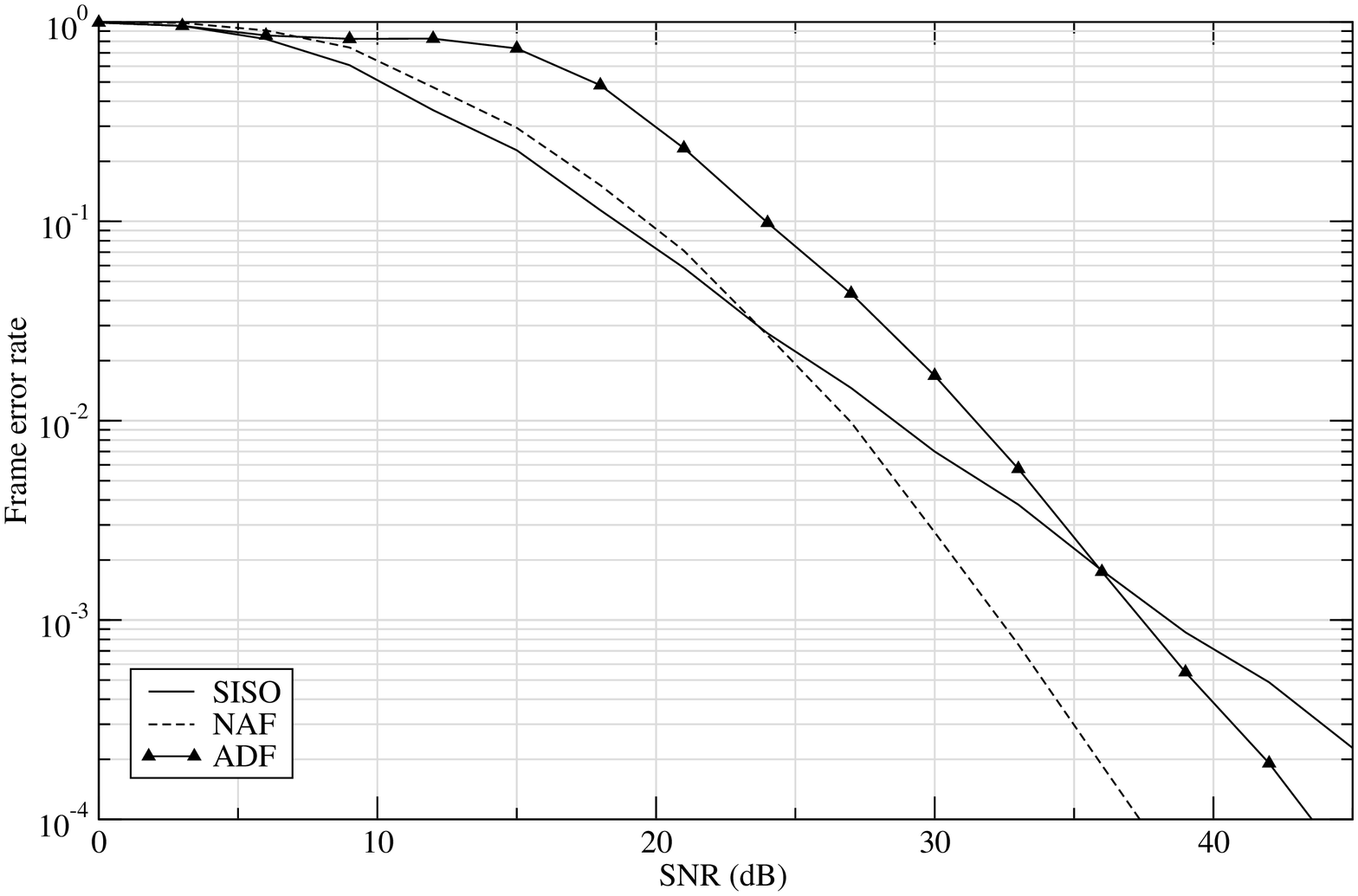}
\caption{Frame error rate of the SISO, NAF and Asymmetric DF protocols (both implemented with the Golden code) for 1 relay, 4 bits pcu} \label{wer_as_df}
\end{figure}

\begin{figure}[h!t]
\centering
\includegraphics[width=\linewidth,clip]{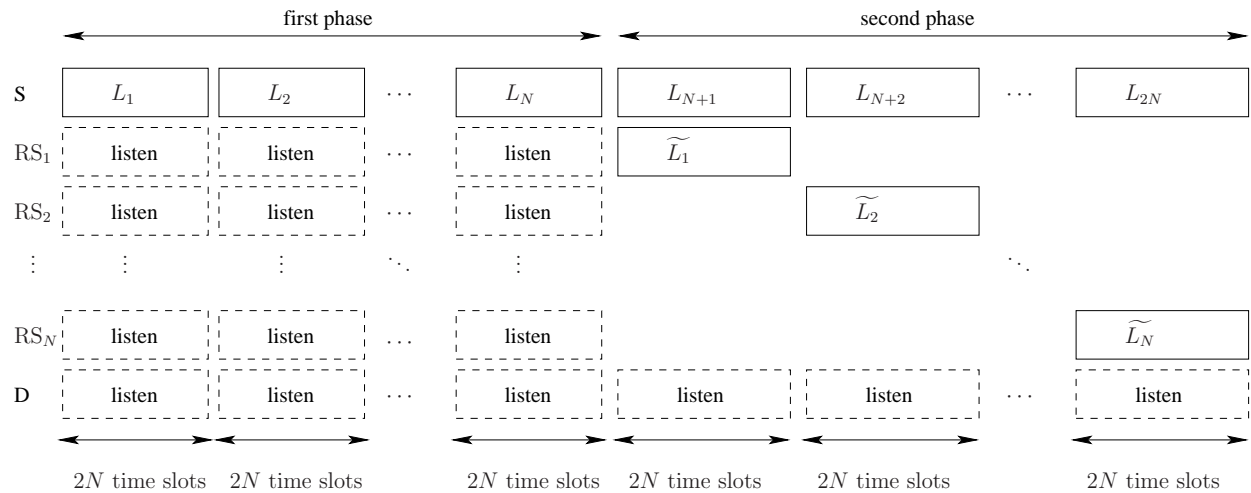}
\caption{Transmission frame of the Incomplete DF protocol} \label{TRANS_SCHEME2}
\end{figure}

\begin{figure}[h!t]
\centering
\includegraphics[width=.7\linewidth,clip]{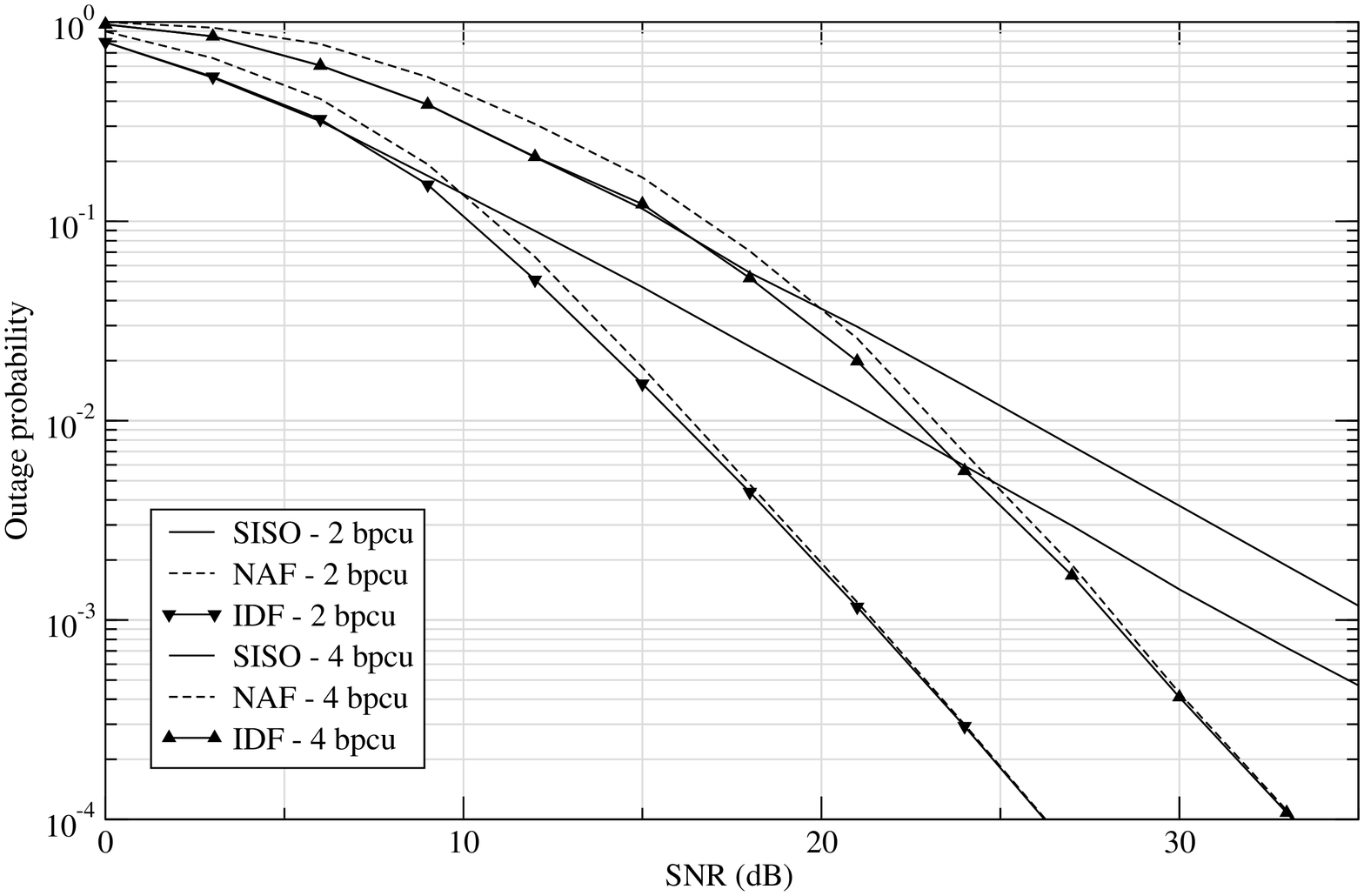}
\caption{Outage probabilities of the SISO, NAF and new DF protocols as functions of the SNR at spectral efficiencies of 2 and 4 bits per channel use in the 1-relay case} \label{pout_1relai}
\end{figure}

\begin{figure}[h!t]
\centering
\includegraphics[width=.7\linewidth,clip]{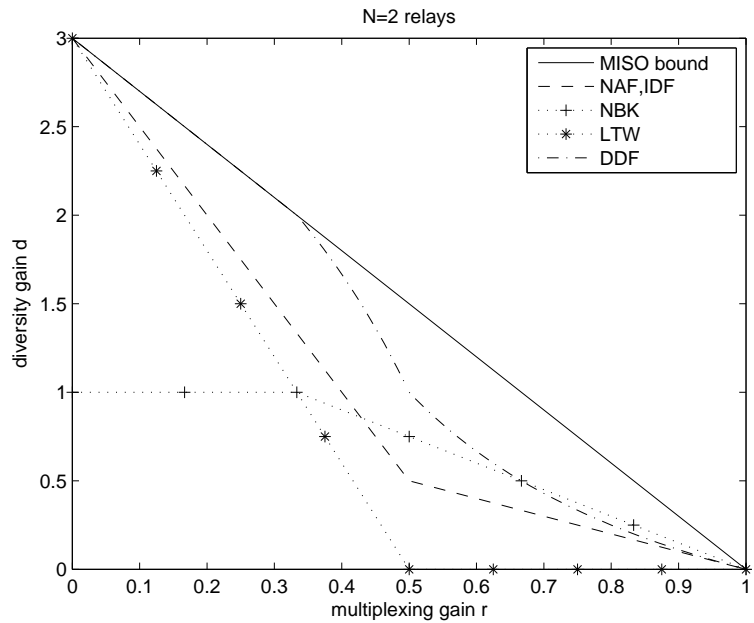}
\caption{DMT of several cooperation protocols} \label{dmt}
\end{figure}

\begin{figure}[h!t]
\centering
\includegraphics[width=.6\linewidth,clip]{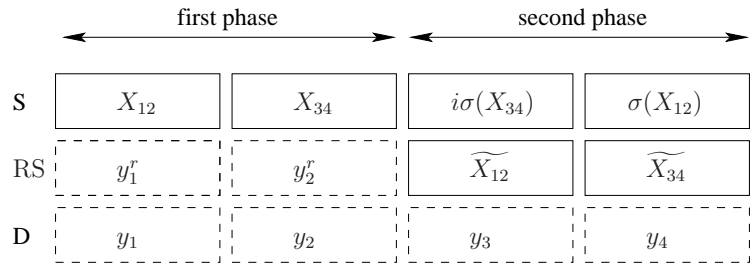}
\caption{Transmission frame of the Incomplete DF protocol in the 1-relay case implemented with a distributed Golden code} \label{GOLDEN_DF}
\end{figure}

\begin{figure}[h!t]
\centering
\subfigure[Constellation $C$]{\includegraphics[width=.5\linewidth]{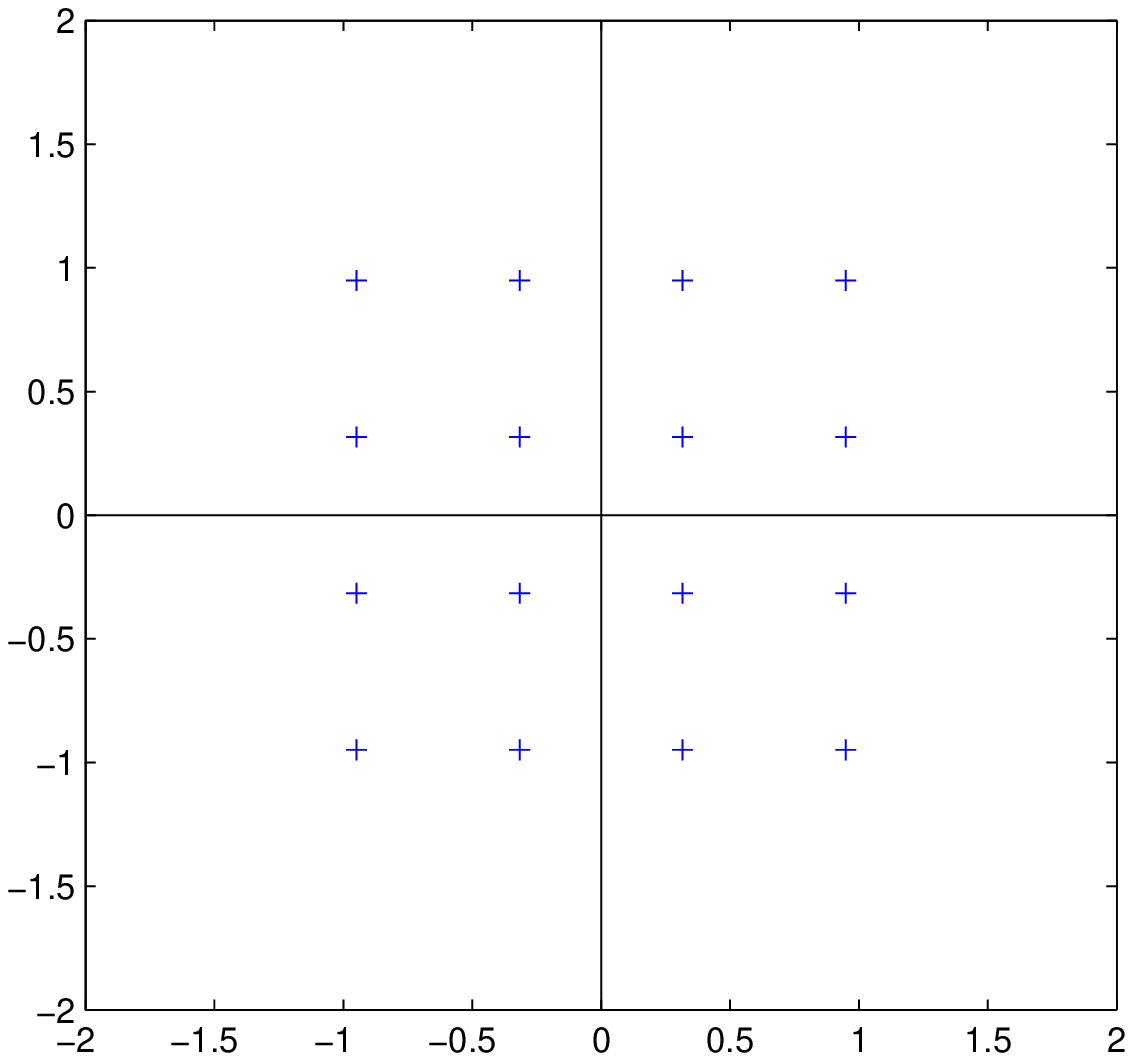}}
\hfil
\subfigure[Constellation $C'$]{\includegraphics[width=.47\linewidth]{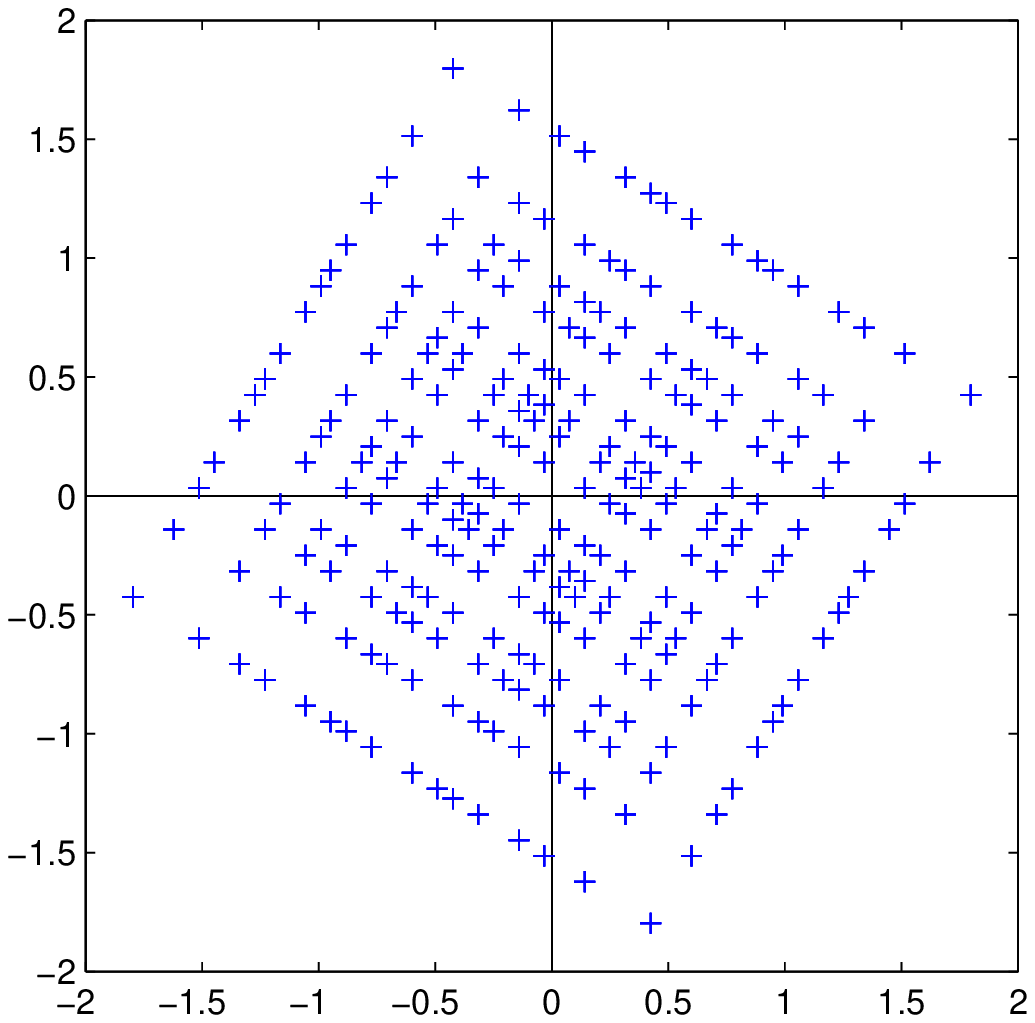}}
\caption{16-QAM and Golden constellation}\label{const}
\end{figure}

\begin{figure}[h!t]
\centering
\includegraphics[width=.7\linewidth,clip]{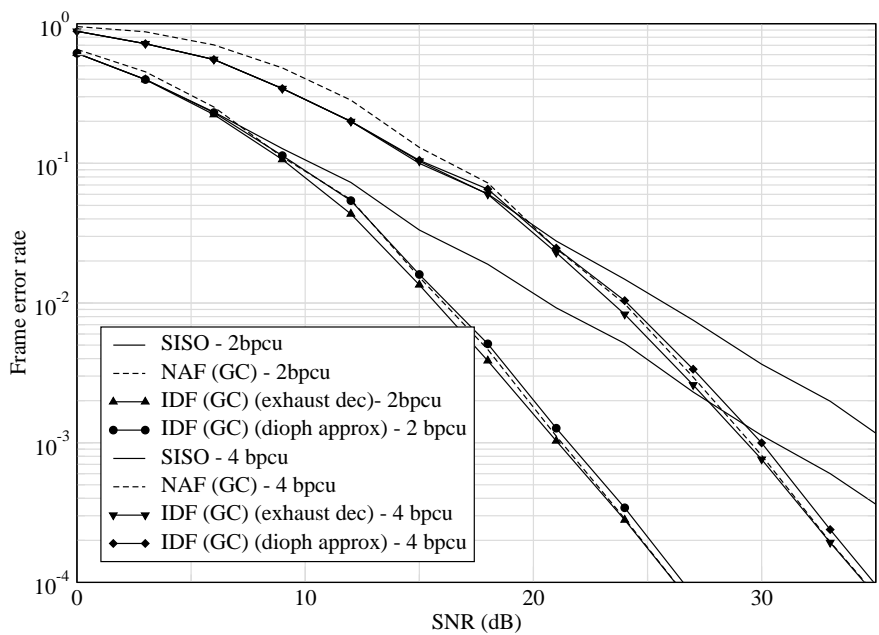}
\caption{Frame error rates of the SISO, NAF and new DF protocols as functions of the SNR at spectral efficiencies of 2 and 4 bits per channel use in the 1-relay case} \label{wer_1relai}
\end{figure}

\begin{figure}[h!t]
\centering
\includegraphics[width=\linewidth,clip]{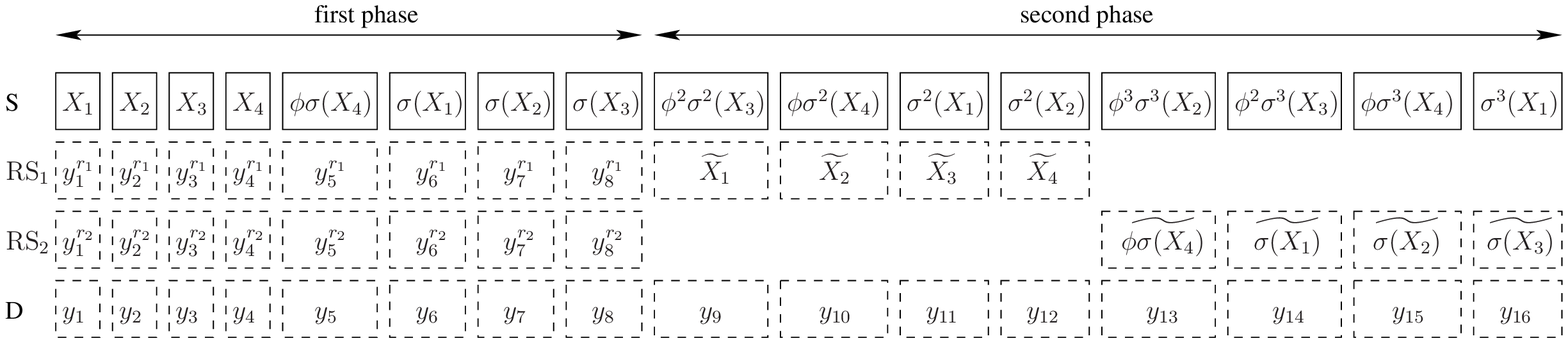}
\caption{Transmission frame of the Incomplete DF protocol in the 2-relay case implemented with a distributed $4\times4$ TAST code} \label{TAST_DF}
\end{figure}

\begin{figure}[h!t]
\centering
\includegraphics[width=.7\linewidth,clip]{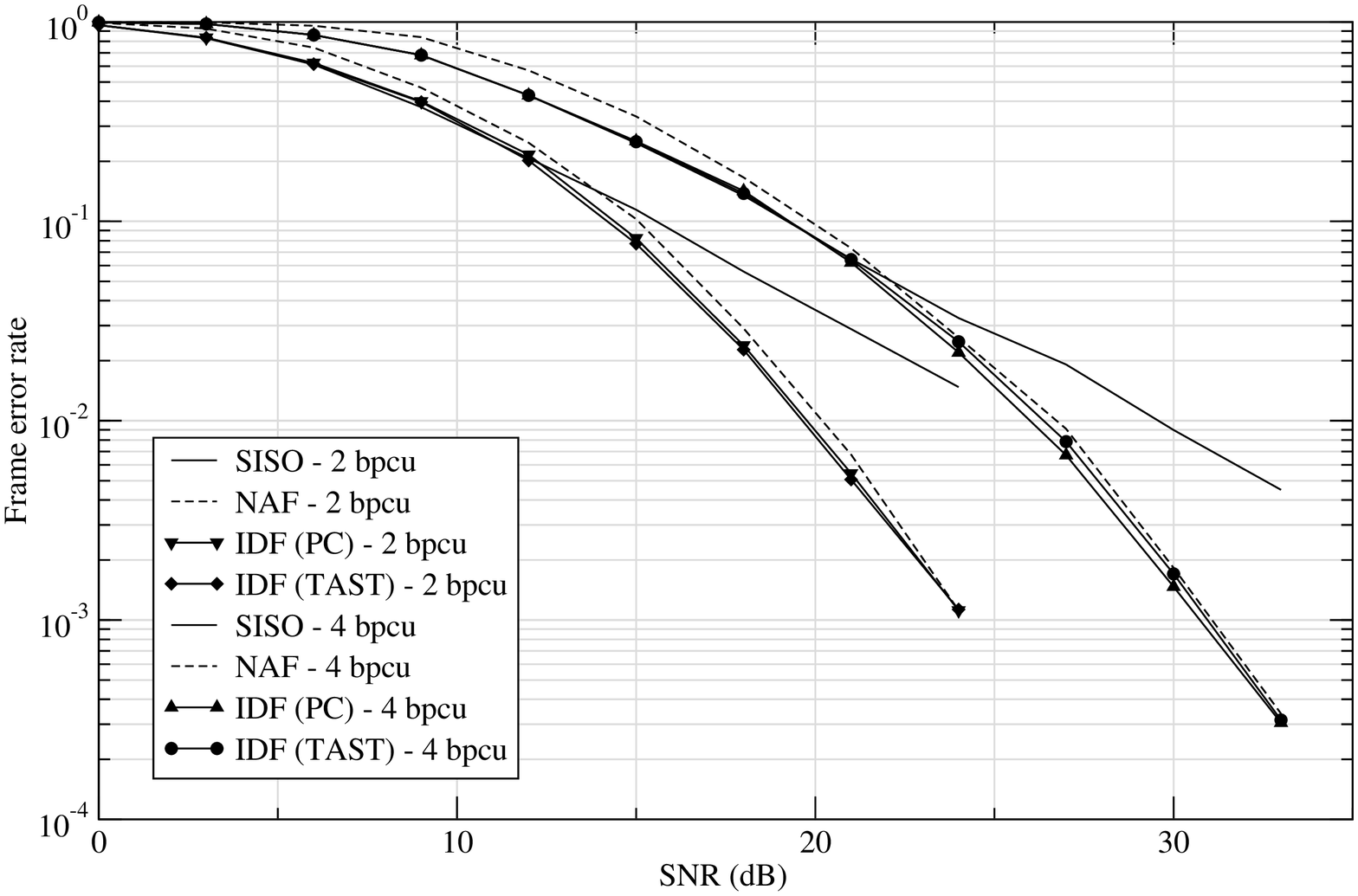}
\caption{Frame error rates of the SISO, NAF and new DF protocols as functions of the SNR at spectral efficiencies of 2 and 4 bits per channel use in the 2-relay case} \label{wer_2relais}
\end{figure}

\end{document}